\documentclass{article}
\pdfoutput=1

\PassOptionsToPackage{numbers, compress}{natbib}

\usepackage[preprint]{neurips_2019}

\usepackage[utf8]{inputenc} 
\usepackage[T1]{fontenc}    
\usepackage{hyperref}       
\usepackage{url}            
\usepackage{booktabs}       
\usepackage{amsfonts}       
\usepackage{nicefrac}       
\usepackage{microtype}      

\usepackage{graphicx}
\usepackage{subfigure}

\usepackage{amsmath}
\usepackage{amssymb}
\usepackage{mathtools}

\usepackage{algorithm}
\usepackage{algorithmic}
\usepackage{hyperref}

\usepackage{amsmath,amsfonts,amssymb,amsthm,epsfig,epstopdf,titling,url,array}

\theoremstyle{plain}
\newtheorem{thm}{Theorem}[section]
\newtheorem{lem}[thm]{Lemma}

\newtheorem*{cor}{Corollary}

\theoremstyle{definition}
\newtheorem{defn}{Definition}[section]

\theoremstyle{remark}

\usepackage{tikz}
  \def\firstcircle{(90:.55cm) circle (.75cm)}
  \def\secondcircle{(210:.55cm) circle (.75cm)}
  \def\thirdcircle{(330:.55cm) circle (.75cm)}

\title{Probabilistic Partitive Partitioning (PPP)}

\author{
  Mujahid Sultan \\
  Department of Computer Science\\
  Ryerson University\\
  \texttt{mujahid.sultan@ryerson.ca} \\
}

\begin{document}

\maketitle

\begin{abstract}
Clustering is NP-hard problem. Thus, no optimal algorithm exists, heuristics are applied to cluster the data. Heuristics can be very resource intensive, if not applied properly. For substantially large data sets computational efficiencies can be achieved by reducing the input space, if minimal loss of information can be achieved. Clustering algorithms, in general, face two common problems: 1) these converge to different settings with different initial conditions and; 2) the number of clusters has to be arbitrarily decided beforehand. This problem has become critical in the realm of big data. Recently, clustering algorithms have emerged which can speedup computations using parallel processing over the grid but face the aforementioned problems. \textbf{Goals:}  Our goals are to find methods to cluster data which: 1) guarantee convergence to the same settings irrespective of the initial conditions; 2)  eliminate the need to establish number of clusters beforehand, and 3) can be applied to cluster large datasets. \textbf{Methods:} We introduce a method which combines probabilistic and combinatorial clustering methods to produce repeatable and compact clusters which are not sensitive to initial conditions. This method harnesses the power of k-means (a combinatorial clustering method) to cluster/partition very large dimensional datasets, and uses Gaussian Mixture Model (a probabilistic clustering method) to validate the k-means partitions. \textbf{Results:} We show that this method produces very compact clusters which are not sensitive to initial conditions. This method can be used to identify the most 'separable' set in a dataset which increases the 'clusterability' of a dataset. This method also eliminates the need to specify number of clusters in advance.
\end{abstract}
\section{Introduction} 
 
Clustering is a process of partitioning a dataset $D$ into $k$ subsets called clusters  $D_i, i=1,\ldots,k$, such that some distance measure is minimised within clusters and maximised between clusters. Shai et al \cite{ackerman2009clusterability} identified that it is difficult to achieve these goals together, and  \cite{balcan2009approximate, ostrovsky2006effectiveness, ostrovsky2012effectiveness} showed that if the data is well-clusterable according to a certain ''clusterability'' or ''separation'' condition then various Lloyd  \cite{lloyd1982least} style methods do indeed perform well and return a provably near-optimal clustering. We present a method which discovers the best ''seperation'' in a dataset to increase its ''clusterability''. A comprehensive overview of clustering methods, state-of-the-art, issues and empirical studies can be found in \cite{berkhin2006survey} and \cite{xu2005survey}. A detailed overview of clustering methods for \emph{large datasets} is given in \cite{fahad2014survey}.

Clustering can be divided into two classes: agglomerative (or bottom-up) and divisive (or  top-down). Agglomerative clustering approaches are inherently inferior in defining clusters, as the clustering decisions are made at root level and any mistake made earlier cannot be corrected later on~\cite{friedman2001elements}. Divisive or top-down clustering approaches produce better results because these consider full populations for making cluster decisions to begin with. One of the most common and well studied divisive clustering algorithm is k-means \cite{wu2012advances, mangiameli1996comparison}. Despite its wide spread use and ability to clusters very high dimensional datasets, it has two major drawbacks: 1) sensitivity to the initial conditions and 2) the need to specify the number of clusters in advance \cite{arthur2007k}.
Formally, the k-means problem can be stated as follows. Given an integer $k$ and a set of $n$ data points $x \in \mathbb{R}^d$, choose $k$ centres $\textsc{c}$ so as to minimise the objective function $f$,
$f = \sum_{x \in X} min_{c \in \textsc{c}} ||x-c||^2.$
Combinatorial algorithms, like k-means, are based on iterative greedy descent. An initial partition is calculated. At each iterative step, the cluster assignments are revised in such a way that the value of some criterion (e.g. similarity or distance measure) is improved from its previous value. This type of clustering algorithms alter cluster assignments at each iteration. When the criterion (e.g., similarity or distance measure) is unable to improve, the algorithm is terminated with the current assignments as its clustering structure. In very high dimensional spaces, these algorithms converge to local optima/local minima which may be highly sub-optimal when compared to the global optimum \cite{friedman2001elements}. Lloyd's \cite{lloyd1982least} algorithm is widely used implementation for k-means. The reason of its popularity is its simplicity and speed (with which it converges). 

\subsection{Background}
The need to design PPP primarily arose from the problem sets which involve clustering of both the feature space and the instance space, such as microarray gene expression clustering \cite{sturn2002genesis}, in which the features are tissue samples (or the patients) and the instances are the genes. It is very important to group similar or discriminating genes for a specific disease (sample) as well as it is equally important to group the samples/patients. We can think of it as a  two-way clustering problem. Portfolio of stocks, for a specific period of time, is another example, with historical values (of a stock) as instances. Feature selection in a large dataset of images is another example, where images are instances and features to be clustered as lateral dimension. We define the design matrix in the Section~\ref{desig_mat}.

K-means implementation by Loyd's \cite{lloyd1982least} algorithm and similar implementations like k-means++ \cite{arthur2007k} have an inherent issue, that these are poised to get stuck in a local minima while clustering very high dimensional spaces. 

Dimensionality has to be reduced in instance space to conduct a meaningful analysis. Some of the most used methods for dimensionality reduction are Principle Component Analysis (PCA) \cite{jolliffe2002principal}, Locally Linear Embedding (LLE) \cite{roweis2000nonlinear}, and Support Vector Machines (SVM)~ \cite{cristianini2000introduction, tayal2014primal}. The prime task of dimensionality reduction methods is to preserve a global structure of the data as much as possible. PCA and SVM take a projection of the data to a suitable or favourable dimension \cite{roweis2000nonlinear}.  The resulting dimensionality reduction may help the task at hand like visualisation. Unfortunately, these methods are not good for compression of the data,  when the re-construction of original data is required \cite{roweis1998algorithms}. Vector Quantization (VQ), on the other hand, is a much better mechanism for dimensionality reduction of high dimensional data, if re-construction is needed \cite{gersho2012vector}.

\section{Methods: Probabilistic Partitive Partitioning (PPP)} \label{PPP}
The idea is to use k-means in the feature space to cluster the dataset, k-means is quite fast and handles high dimensionality very well. But the issue with the k-means is that it can very easily get stuck in a local minima when the data is very high dimensional. PPP offers a mechanism to take k-means out of local minima. 
Arthur and Vassilvitskii~\cite{arthur2007k} prescribed k-means++, which optimises k-means initial conditions so that it does not get stuck in a local minimum. But this method will not work well if applied to very high-dimensional feature space. Bahamani et al.~\cite{bahmani2012scalable} described a method to parallelize k-means++, but the method has limitations when the data are distributed across the network. On the other hand, PPP works well in very high-dimensional feature space and can deal with distributed data over the network as well.The PPP algorithm is described below.

The pseudo code of PPP is given in the Algorithm \ref{alg:PPP}. In section~\ref{PPP_details} we describe each step of the algorithm in detail.
\subsection{The Design Matrix} \label{desig_mat}
Here we define the design matrix and dataset of focus. The design matrix $\textbf{x} \in \mathbb{R}^{N \times f}$ has $N$ rows or instances and $f$ columns or features. $\Vec{x_i}$ = $x_1,x_2,\ldots,x_f$; ($\Vec{x_i}$ is a $f$ dimensional vector) and $i = 1,\ldots,N$. The instance space is characterised by very high dimensionality: $N \gg 10^7$. The feature space is also vast: $f \sim 10^2 - 10^3$. 
\begin{algorithm}[ht!]
\caption{PPP Algorithm}\label{alg:PPP}
\begin{algorithmic} 
\STATE \textbf{Step 1:} Perform VQ on instance space
\STATE \textbf{Step 2:} Build Gaussian Mixture Model (GMM) on root coodebook vectors
\WHILE{\textbf{$\Phi$ is NOT optimal}}
\STATE \textbf{Step 3:} Partition sample space (entire data or GMM vectors with probabilities above .5) using k-means with C=2
\STATE \textbf{Step 4:} Perform VQ on both partitions
\STATE \textbf{Step 5:} Build GMMs for each partitions from the coodebook vectors
\STATE \textbf{Step 6:} Find posterior probabilities of root coodebook of generating the partition (by k-means)
\STATE \textbf{Step 7:} Update PPP objective function $\Phi$
\ENDWHILE
\STATE \textbf{Step 8:} Repeat steps 1 to 7 for each partition
\end{algorithmic}
\end{algorithm}

\subsection{Detailed Description of PPP} \label{PPP_details}

\subsubsection{Step 1: VQ of instance space} 

The first step of the PPP algorithm reduces the instance space with vector quantization (VQ) using Self-Organising Maps (SOM) \cite{kohonen1998self, kohonen2013essentials}. We use Matlab implementation of SOM Toolbox \cite{vesanto1999self}. Just like scalar quantization is used in telecommunications, where a continuous signal is reduced to quantized signal and transmitted across the network, and re-constructed back at the other end (this is called coding and encoding). Similarly VQ is a mapping of the data vectors from input space to a reduced space, called codebook vectors. Comprehensive background and theory of VQ is given by \cite{gray1998quantization}.

VQ methods like Dirichlet tessellations \cite{dirichlet1850reduction} and Voronoi tessellations \cite{voronoi1908nouvelles} date back to ninetieth century. And in the modern times introduced by Lloyd's \cite{lloyd1982least}  (in scalar form) and Forgy \cite{forgy1965cluster} in vector form. 

VQ is defined as $m_{k} \leftarrow \textbf{x}; k = 1,\ldots,K,$ where $m_k$ is randomly initiated codebook vector the numbers of which is either arbitrarily selected (to represent the granularity required) or based on some measure like principle components of the data. VQ objective is to find $m_c$, which is the winning vector given by
\begin{equation}\label{VQ}
\|\textbf{x}-m_{c}\| =   min_{\substack{i}} { \|\textbf{x}-m_{k}\| }
\end {equation}
and the mean quantization error is given by $E =  \int  \| \textbf{x} - m_{c} \|^{2} p(\textbf{x}) \mathrm{d}V$ using Euclidean distance. 
where, $p(\textbf{x})$ is the probability density of the data and $ \mathrm{d}V$ volume differential. $E$ is minimised by gradient descent.
Now at any time $t$ let, $ m_{c} = m_{c}(t)$ and $x_i = x_i(t)$ 
The gradient descent optimisation in the $m_{c}$ space is given by
\begin{equation} \label{eq:som_learn}
m_{c}(t+1) = m_{c}(t) + \alpha (t) \left[ x_i(t)  - m_{c}(t) \right],
\end{equation}
\[
m_{k}(t+1) = m_{k} ; k\neq c,
\]
here $\alpha(t)$ is monotonically decreasing sequence of scalar valued gain coefficients $0\leq \alpha(t) \leq 1$.

A simple architecture of SOM  is shown in Figure \ref{fig:som_architecture}(a). SOM uses a kernel function, called neighbourhood function $N_c$ and shown in Figure \ref{fig:som_architecture}(a). At each learning step all the cells within the neighbourhood $N_{c}$ are updated where as the cell out side the $N_{c}$ are left intact. The most commonly used neighbourhood function $N_c$ is Gaussian, given by 
\[h(c, i) = \alpha (t) exp[-sqdist(c,i)/2\sigma^2(t)],\] where $\sigma$ is monotonically decreasing function of time, $sqdist(c, i)$ is the square of the geometric distance between the nodes $c$ and $i$ of the grid of SOM components.

This neighbourhood is found around a cell which is the best match to a data point (vector) in the instance space. $N_c$'s radius is decreased monotonically with the time as shown in Figure \ref{fig:som_architecture}(a). Now replace $\alpha(t)$ with the kernel so the Equation~\ref{eq:som_learn} can be written as: 
\begin{equation} \label{eq:som_learn3}
m_{k}(t+1) = m_{k}(t) + h(t, c, i) \left[ x_i(t)  - m_{k}(t) \right],
\end{equation}
$\because h_{ci}(t)=\alpha(t)$ within $N_{c}$ and $h_{ci}(t) =0$ outside. 
where $c$ is the best matching unit.

\begin{figure}[ht]
\centering
(a)
    \includegraphics[height=3.5cm]{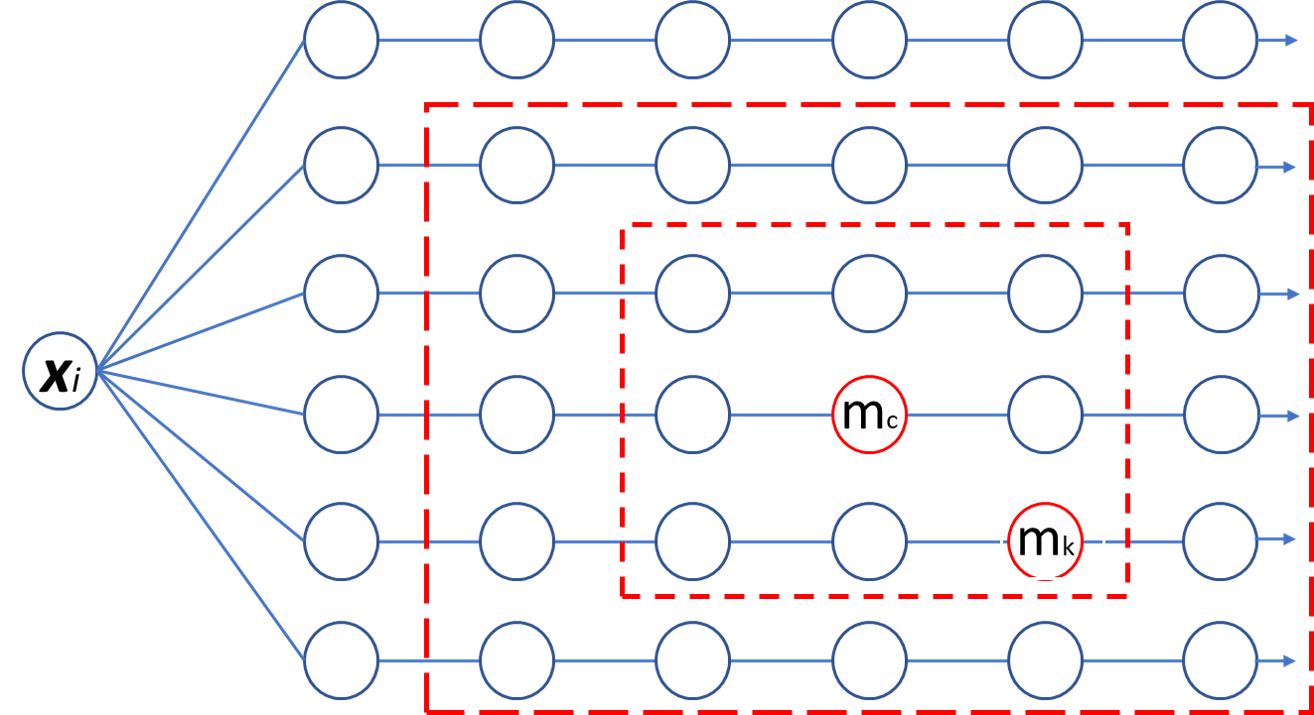}
(b)
    \includegraphics[height=3.5cm]{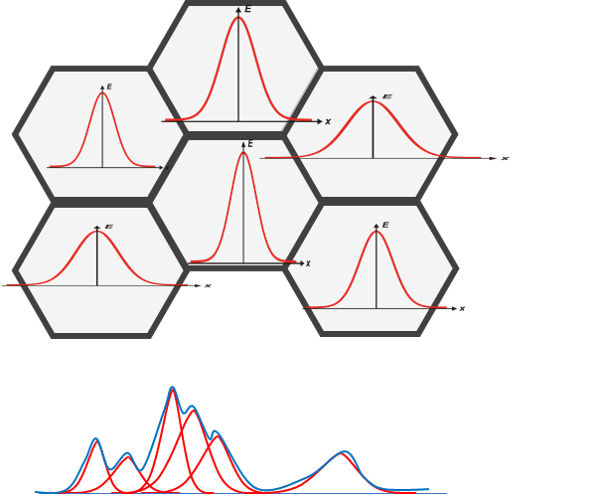}

    \caption{\label{fig:som_architecture} (a) SOM architecture and neighbourhood kernel function $h_{ci}$. (b) Mixture of Gaussian on SOM components} 
\end{figure}

At each iteration of the SOM,  the codebook vectors are updated as per the neighbourhood function, and a best matching unit, a SOM component, is identified to a randomly selected data point from $\textbf{x}$. This operation is repeated until all data points in $\textbf{x}$ are assigned to codebook vectors $m_k$. This clusters most of the dataset $\textbf{x}$ into $K$ bins or components (note the data vectors with high quanitzation error will not have any bin assigned).
\begin{defn}
Vectors in the instance space with least VQ error (from root level codebook vector set, $m_0$) is defined as $x_{m_{0}}$; ($x_{m_{0}}$ is $K \times f$ dimensional)
\end{defn}

\subsubsection{Step 2: Generate root level Gaussian Mixture Model} 
After the map converges (in the previous step) and all the codebook vectors are updated with similar vectors from $\textbf{x}$. We assume that each data vector in $x_{m{_0}}$ can be approximated by a Gaussian component.
\begin{defn}
At the root level Gaussian mixture model has $k$ components, each vector in set $x_{m_{0}}$ is center of a Gaussian component. $\mathbb{G}_{mm_0}$ is the probability density function (pdf) of each data vector in $\Vec{x_i}$, of being generated by this mixture model.
\end{defn} 

We generate a Gaussian Mixture Model of these components, $\mathbb{G}_{mm_0}$, utilizing generative modelling properties of Gaussian Mixtures as described below. We estimate parameters of the mixture model using  EM algorithm, and describe (in next steps) how these parameters help validate the clustering done by k-meas in feature space.
\begin{defn}
The probability density of each component $m_k$ in the mixture model is defined as
\begin{equation}\label{GMM}
p(m_k) = \sum_{k=1}^{K} \pi_k \mathcal{N}(\textbf{x}|\mu_k, \Sigma_k),
\end{equation}
where $\pi_k$ is mixing coefficient of each component in the mixture, such that $\sum_{k=1}^k (\pi_k =1); \forall{\pi_k >= 0}$.
\end{defn}

The task is to find probability of each data vector in $\Vec{x_i}$ with respect to the entire mixture. The assumption is that the data is generated by selecting  probability $\pi_k$ as mixture component $\theta_k (\theta_k =\pi_k,\mu_k, \Sigma_k)$ and then drawing a data item from the corresponding distribution $p(.|\theta_k)$.

Given data vectors in each component $m_k$ and parameters $\theta_k$ for each component, the EM algorithm is used to maximize $\theta_k$ of the mixture model. 


The task is to maximize Equation \ref{GMM}, which is done by taking log and differentiating with respect to $\mu$ and $\Sigma$. Log is used, as maximising log is equivalent of maximizing a function and it offers computation ease (takes care of division by zero).

The log likelihood of the entire mixture model is given by:
\begin{equation}\label{Mixture_ll}
\ln p(\textbf{x}|\theta_k) = \ln \sum_{i=1}^{N} p(m_k)
\end{equation}
and with respect to each component:

\begin{center}
$= \sum_{i=1}^{N} \ln \left( \sum_{k=1}^K \pi_k  \mathcal{N} (\textbf{x}|\mu_{k}, \Sigma_k) \right)$    
\end{center}
where $\mathcal{N} (\textbf{x}|\mu_{k}, \Sigma_k)$ 
of each component is calculated by maximum likelihood estimate for normal distributions, given by:
\begin{equation}\label{logLL1}
\begin{split}
\ln p(\theta_k) &= \sum_{i=1}^{N} \ln \pi_k  + \ln \bigl[(2\pi)^{f} |\Sigma|^{-1/2} \\
& -  (x_{i} - \mu_k)^{T}\Sigma^{-1}(x_i-\mu_k) / 2 \bigr]. \\
\end{split}
\end{equation}
Dropping the constant additive terms in \eqref{logLL1} we get the log-likelihood as
\begin{equation}\label{logLL}
\ln p(\theta) = \frac{1}{2} \ln |\Sigma| ^{-1/2} - \frac{1}{2} \sum_{i=1}^{N} (x_{i} - \mu_{k})^{T}\Sigma^{-1}(x_{i}-\mu_{k})
\end{equation}

Setting equation \ref{logLL} to zero and differentiating with respect to  $\mu$ and $\Sigma$ gives us:
\begin{equation}\label{mu_ml}
\mu_{ml_k} = \sum_{i=1}^{N} x_{i}
\end{equation}
\begin{equation}\label{sigma_ml}
{\Sigma_{ml_k}} = \frac{1}{N}\sum_{i=1}^{N} (x_{i} - \mu_k)(x_{i} - \mu_k)^{T},
\end{equation}
which are called sufficient statistics for mixture models, or Maximum Likelihood (ML) estimates of a Mixture model. The solution of Equation \ref{logLL} can not be found in closed form, therefore EM algorithm is used. We use $\mu_{ml_k}$ and $\Sigma_{ml_k}$ of each component $k$ to calculate probabilities of each data point in $\Vec{x_i}$ of being generated by this Mixture model, called $\mathbb{G}_{mm_0}$. A schematic diagram of linear superposition of Gaussians, from the components of SOM's is shown in  Figure \ref{fig:som_architecture}(b).  

\begin{defn}
$\Gamma_0$ is the index of vectors in instance space with $\mathbb{G}_{mm_{0}} > .5$
\end{defn}

\textbf{While loop start:}

\subsubsection {Step 3: Partition sample space using k-means with C=2} 
Partition the feature space using standard k-means. We can partition the entire set $\textbf{x}$ or $\Gamma_0$. k-means converges quite fast even for very large dimensions, but might get stuck in a local minima. Step 4, takes k-means out of the local minima as explained below. We use $C=2$ at each iteration, this enables building a binary-tree. At fist step we partition $\textbf{x}$ or $\Gamma_0$ into two sets $\textbf{x}_1$ and $\textbf{x}_2$, such that $\textbf{x}_1 \in \mathbb{R}^{N \times f1}$ and $\textbf{x}_2 \in \mathbb{R}^{N \times f2}$; where $f = f_1 + f_2$.

\subsubsection{Step 4: Vector quantization of each child} 
Repeat the step~1 for both $\textbf{x}_1$ and $\textbf{x}_2$ by randomly initiating two codebook vector sets $m_{11}$ and $m_{12}$. We identify the closest ids of observations for each codebook centeroids $m_{11}$ and $m_{12}$ from the children $\textbf{x}_1$ and $\textbf{x}_2$. 

\begin{defn}
The id's of closest observations of parent dataset from the converged codebook $m_0$ is defined as $x_{m_0}$ and ids of the closest observations for children as $x_{m_{11}}$ and $x_{m_{12}}$. 
\end{defn}

\subsubsection{Step 5: Generate Gaussian Mixture Model for each child} 

Repeat the same process as described in the Step\#2 for each partitioned dataset $\textbf{x}_1$ and $\textbf{x}_2$, where we randomly generate two coodebook vectors $m_{11}$ and $m_{12}$ and learn two mixture models, one from the dataset $x_{m_{11}}$ and other on the dataset $x_{m_{12}}$ and  learn parameters of each mixture model using EM algorithm, giving us  $\mathbb{G}_{mm_{11}}$ and $\mathbb{G}_{mm_{12}}$. 

\subsubsection{Step 6: Find Posteriors of root codebook in generating each child mixture models}

To find the posterior probabilities of root codebook $x_{m_{0}}$ to be responsible of generating child Gaussian mixtures, we use joint probability of root and child datasets, and find the posteriors using Bayes theorem, which is stated as $p(y|x) = \frac{p(x|y)p(y)}{\sum p(x)} $.

\begin{defn}
Posteriors of set $x_{m_{0}}$ in generating child mixture models are defined as $p(x_{m_{0}}|\mathbb{G}_{mm_{11}})$ and $p(x_{m_{0}}|\mathbb{G}_{mm_{12}})$ given by:

\begin{equation}{\label{xm0_gmm_11}}
p(x_{m_{0}}|\mathbb{G}_{mm_{11}}) = \left( \mathbb{G}_{mm_{11}} \odot p(x_{m_{0}}) \right) \oslash \left(\sum_{k=1}^{K} (\mathbb{G}_{mm_{11}})_k \right),
\end{equation}
and
\begin{equation}{\label{xm0_gmm_12}}
p(x_{m_{0}}|\mathbb{G}_{mm_{12}}) = \left( \mathbb{G}_{mm_{12}} \odot p(x_{m_{0}}) \right) \oslash \left(\sum_{k=1}^{K} (\mathbb{G}_{mm_{12}})_k \right)
\end{equation}
where $\odot$ represents element-wise multiplication and $\oslash$ -- element-wise division. 
\end{defn}

Where as, $p(x_{m_{0}})$is apriori probability of each data vector $x_{m_{0}}$ given by the number of hits a SOM component gets multiplied by the neighbourhood function $h_{ci}$ and divided by the sum, given as: 
\begin{equation}{\label{p_m_0}}
p(x_{m_{0}})= \nicefrac {p(m_{0}) \times h_{ci}}{\sum_{k=1}^{K} p(m_{{0})_k} h_{{ci}_k}} 
\end{equation} 

The product of posterior probabilities $p(x_{m_{0}}|\mathbb{G}_{mm_{11}})$ and $p(x_{m_{0}}|\mathbb{G}_{mm_{12}})$ is used as stopping criteria. The larger the product the better the split.
Note - For ease of computations we keep the number of SOM components for each child  $K$ as well. 
\begin{lem}
Posteriors of $x_{m_{0}}$ in generating the child mixture models $\mathbb{G}_{mm_{11}}$ and $\mathbb{G}_{mm_{12}}$ provide a set of vectors in the instance space which are the most discriminating set for partitioning a dataset in the feature space.
\end{lem}

\subsubsection{Step 7: Stopping Criteria and the most discriminating set}
\begin{defn}
In the instance space, $\Gamma_0$ is the index of vectors with $\mathbb{G}_{mm_{0}} > .5$, $\Gamma_1$ is the index of vectors with $p(x_{m_{0}}|\mathbb{G}_{mm_{11}}) > .5$, and
$\Gamma_2$ is the index of vectors with $p(x_{m_{0}}|\mathbb{G}_{mm_{12}}) > .5$
\end{defn}
We compute the fraction of overlap between $\Gamma_1$ and $\Gamma_0$ and between $\Gamma_2$ and $\Gamma_0$ and normalise it by the number of vectors in $\Gamma_1$ and $\Gamma_2$ respectively. These fractions are denoted by $\phi_1$ and $\phi_2$, and are calculated as:
\[\phi_1 = 100 \times \frac{| \Gamma_1 \cap \Gamma_0 |}{|\Gamma_1|},\]
and
\[\phi_2 = 100 \times \frac{| \Gamma_2 \cap \Gamma_0  |}{|\Gamma_2|}.\]
These values are used to compute a function $\Phi$, optimising which gives the best posteriors for both children, this function is defined as:
\begin{equation}{\label{responsibility1}}
    \Phi = \frac{\phi_1 \phi_2}{\phi_1 + \phi_2} 
\end{equation}

The output of this function is shown in Figure~\ref{fig:objective}. If there is no overlap between the parent and children, then  $\Phi$ is undetermined ($0/0$), and the algorithm is terminated (and all the process is repeated again, until $\Phi$ has some positive values), plot of $\Phi$ for different iterations is shown in Figure~\ref{fig:objective}(b). 

The schematic diagram of PPP is shown in Figure~\ref{fig:PPP}(b).
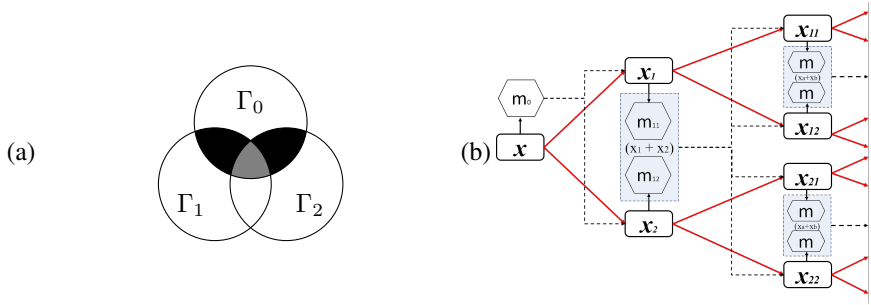
\begin{figure}[ht]
(a)
\minipage{0.4\textwidth}%
\centering
\begin{tikzpicture}
    \begin{scope}
        \clip \thirdcircle;
        \fill[black] \firstcircle;
    \end{scope}
    \begin{scope}
        \clip \secondcircle;
        \fill[white] \thirdcircle;
    \end{scope}
    \begin{scope}
        \clip \secondcircle;
        \fill[black] \firstcircle;
    \end{scope}
    
        \begin{scope}
        \clip \secondcircle;
        \clip \thirdcircle;
        \fill[gray] \firstcircle;
    \end{scope}
    \draw \firstcircle node[text=black,above] {$\Gamma_0$};
    \draw \secondcircle node [text=black,below left] {$\Gamma_1$};
    \draw \thirdcircle node [text=black,below right] {$\Gamma_2$};
\end{tikzpicture}
\endminipage
(b)
\minipage{0.6\textwidth}%
    \includegraphics[height=4cm]{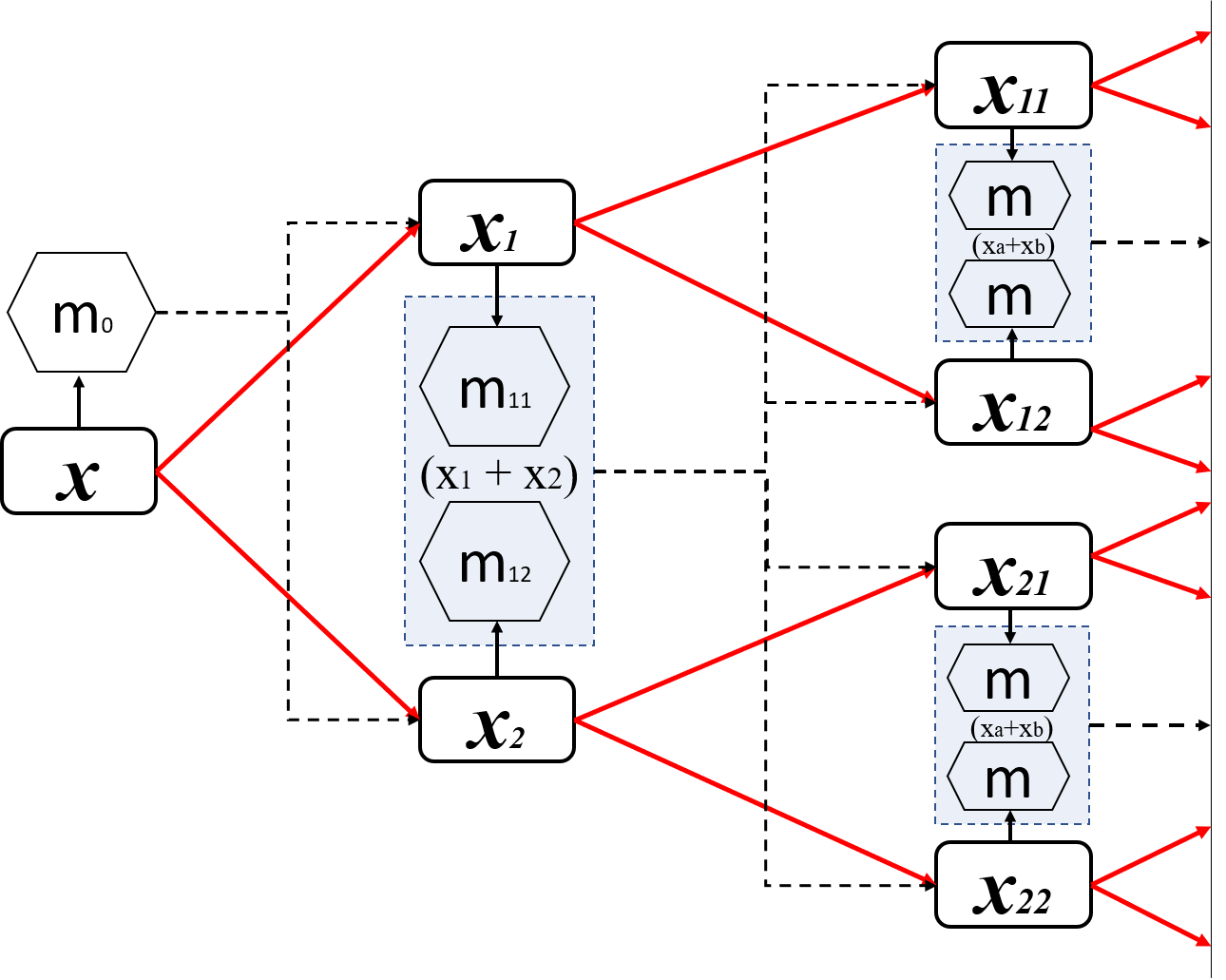}
   
    \endminipage
     \caption{(a) $\Gamma_0$, $\Gamma_1$ and $\Gamma_2$ as defined in step 7 (b) Schematic diagram of PPP}\label{fig:PPP}
\end{figure}
On each node of a tree we iterate till we find this set, and if after some iteration (e.g 20 on the datasets we tried) this set can not found we terminate the split. This builds binary tree with leafs as clusters. To get any desired cluster structure, we can cut this tree at a desired level.

Plot of the posteriors $p(x_{m_{0}}|\mathbb{G}_{mm_{11}})$ and $p(x_{m_{0}}|\mathbb{G}_{mm_{12}})$  is shown in Figure~\ref{fig:objective}(a) (After the second child, we plot average for the subsequent children). Note that these probabilities are max at second split for the dataset 1, and monotonically decrease after that, till a stopping criterion. This indicates that the first split was not that good, and second split was best split (by k-means). As we can imagine, the more we partition a dataset, the more harder it will become to split it further, which is evident from the decreasing responsibilities (posteriors) for both children in Figure~\ref{fig:objective}(a).

\textbf{end while}

\subsubsection{Step 8: Repeat steps 1 to 7 for both partitions} 
Once we split the root dataset into two children and calculate the responsibilities of the parent in generating each child, the vectors in instance space corresponding to $\Gamma_1$ and $\Gamma_2$ are passed on to next level as root set for both partitions of the k-means and steps 1-7 are repeated for both partitions, thus building a binary-tree.
\begin{figure}[ht]

\centering
(a)
    \includegraphics[height=4cm]{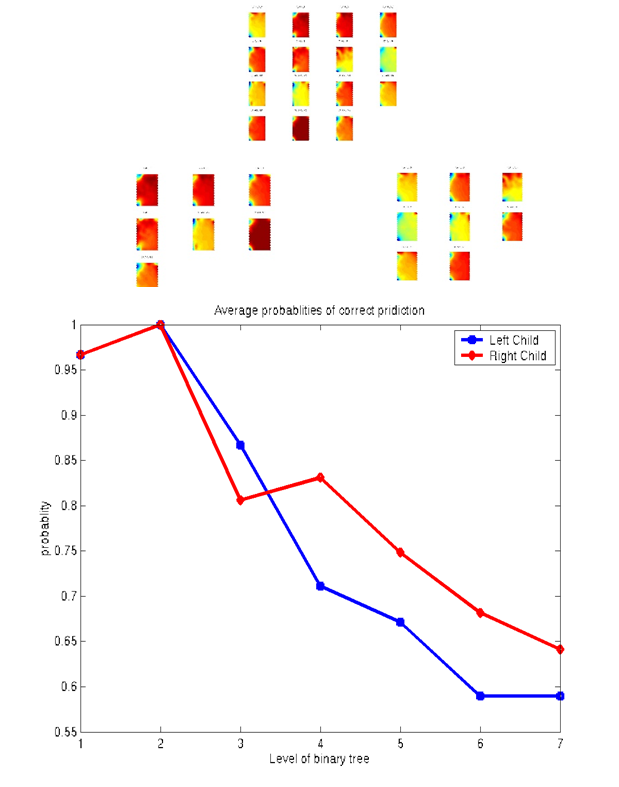}
(b)
     \includegraphics[height=3.5cm]{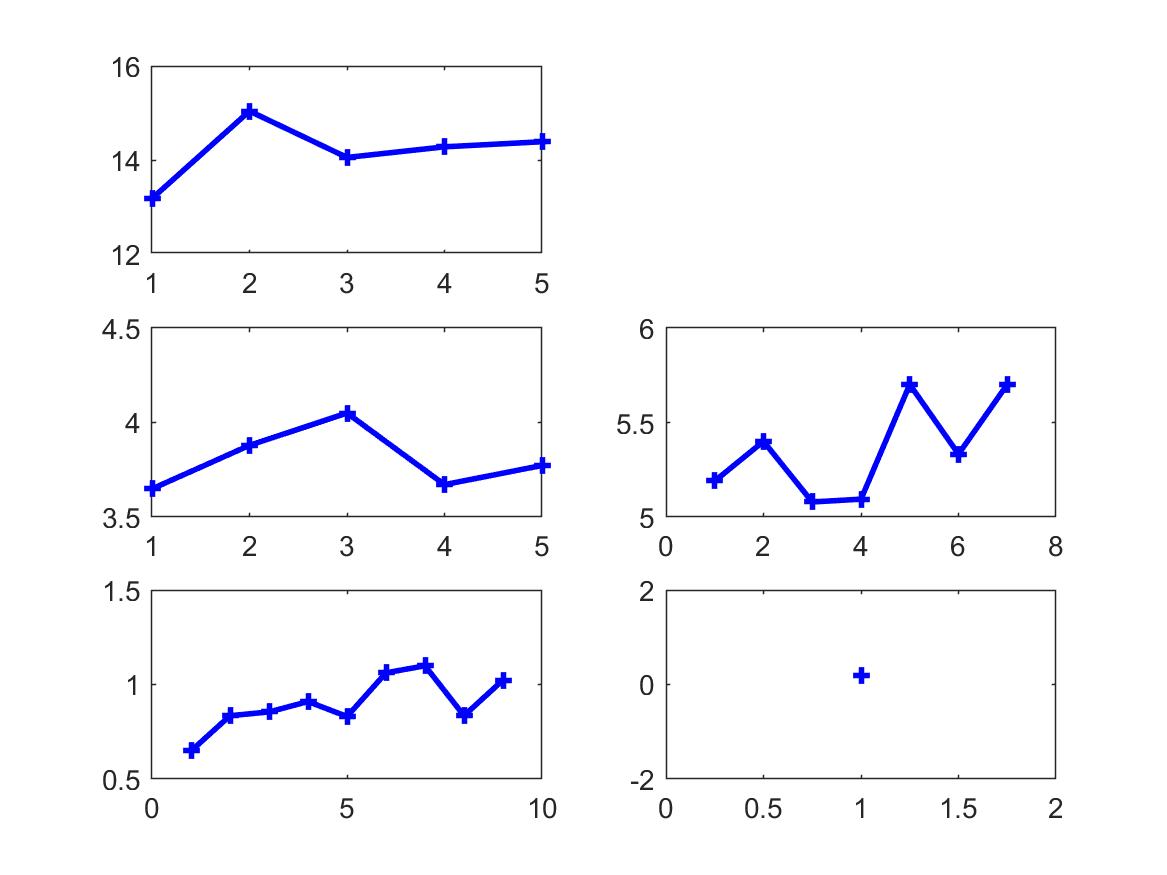}
     \caption{\label{fig:objective} (a) Plot of $\Gamma_1$ and $\Gamma_2$ for dataset-1, (b) Plot of objective function $\Phi$}
\end{figure}
\begin{thm} \label{thm:ppp}
After convergence of $\Phi$, $\Gamma_1$ and $\Gamma_2$ are the most discriminating set of vectors in instance space to split the sample space, and the split by k-manes (in the sample space on these vectors) is best/optimal split. 
\end{thm}

\begin{proof}[Proof of Theorem \ref{thm:ppp}]
If k-means splits the sample space in such a way that the posteriors of the parent GMM in generating the children GMMs are max, this is a best split by the k-means in sample space. This is achieved by modelling joint probabilities of GMMs built on the root and the partitioned sets by k-means. Posteriors of Gaussian Mixtures are cluster boundaries given by 'mahalanobis' distance  $\sum_{i=1}^{N} (x_{i} - \mu_{k})^{T}\Sigma^{-1}(x_{i}-\mu_{k})$. Equations $\ref{xm0_gmm_11}$ and $\ref{xm0_gmm_12}$ are the cluster boundaries in each child. We turn the clustering problem into classification problem by finding the probabilities of $\Gamma_1$ and $\Gamma_2$ to be generated by the parent GMM. $\phi_1$ and $\phi_2$ are the percentages of correct classification of instances present in the parent given a child GMM, and are calculated by generative modelling. 

Bayesian Generative model is built from the GMM of the parent to classify the children. It is intuitive to imagine that if a GMM built on the data set partitioned by k-means has most instances common to the parent GMM, the VQ at the parent and child level has the most common vectors from the instance space. $\Phi$ maximises this set. When $\Phi$ is optimal, the datasets associated with $\Gamma_1$ and $\Gamma_2$ are two distinct clusters in the instance space and partition done by k-means in the sample space is the most optimal partition. 

Optimizing $\Phi$ grantees the same $\phi_1$ and $\phi_2$ for any random initialisation of k-means. The partitioning is terminated when no positive values of $\Phi$ can be found, means there is no data common in the instance space between the VQ of a parent to the VQ of a child. This stops partitioning of the tree branch at a point where no cluster boundaries can be found on child GMM (for posteriors below .5).
\end{proof}
\section{Results and Commentary} 
We analysed few publicly available datasets, below are some results:

\textbf{Dataset-1 }
We took a publicly available contest dataset available \href{http://www.camda.duke.edu/camda03/datasets/index.html}{\underline{here}}. 
This is a microarray gene expression dataset, where genes are organised as rows and columns are different patient samples of the tumour. The cells are ratios of control vs sample cancer. The PPP results are shown Figure \ref{fig:data1}(a). We can see the PPP separated two different type of cancers (this though needs further domain experts verification) but the visual representation of SOM component planes \cite{vesanto1999som} show that PPP has done good clustering. Plus the tree is relatively shallow and the leaf nodes can be treated as clusters. 

\textbf{Dataset-2}
We used a kaggle competition dataset and clustered ALL and AML cancer sub-types of leukemia. The dataset is available  \href{https://www.kaggle.com/crawford/gene-expression}{\underline{here}}.
Though this kaggle competition is for classification (on feature space) we used PPP clustering and the results are presented in Figure \ref{fig:data1}(b). It is evident that the unsupervised learning was able to split two cancers quite well. If we take the leaf sets, it generates the cluster structure and the error rate is quite low.
\begin{figure*}[ht]
\centering
    \includegraphics[height=4cm]{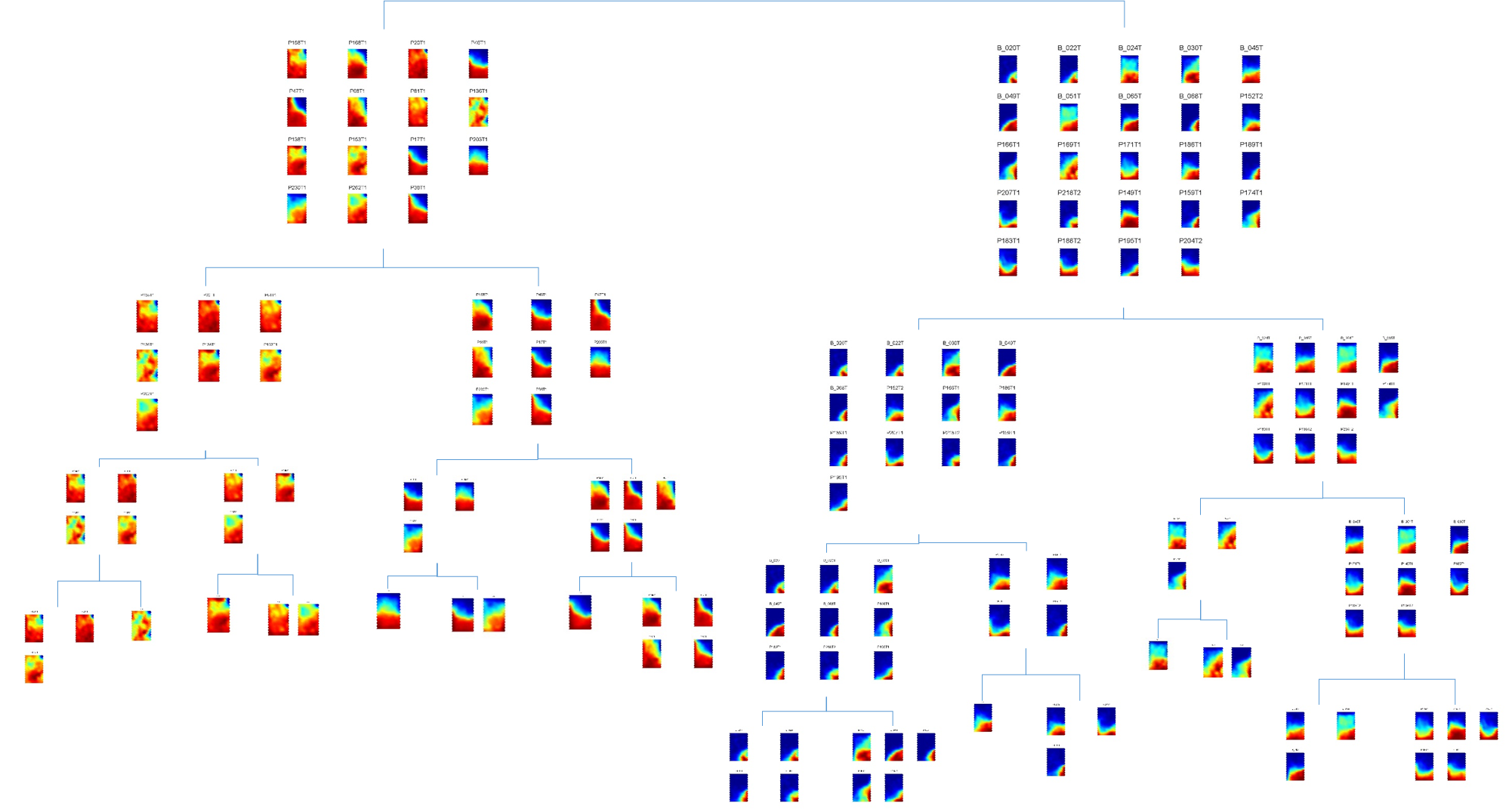}
      \includegraphics[height=4cm]{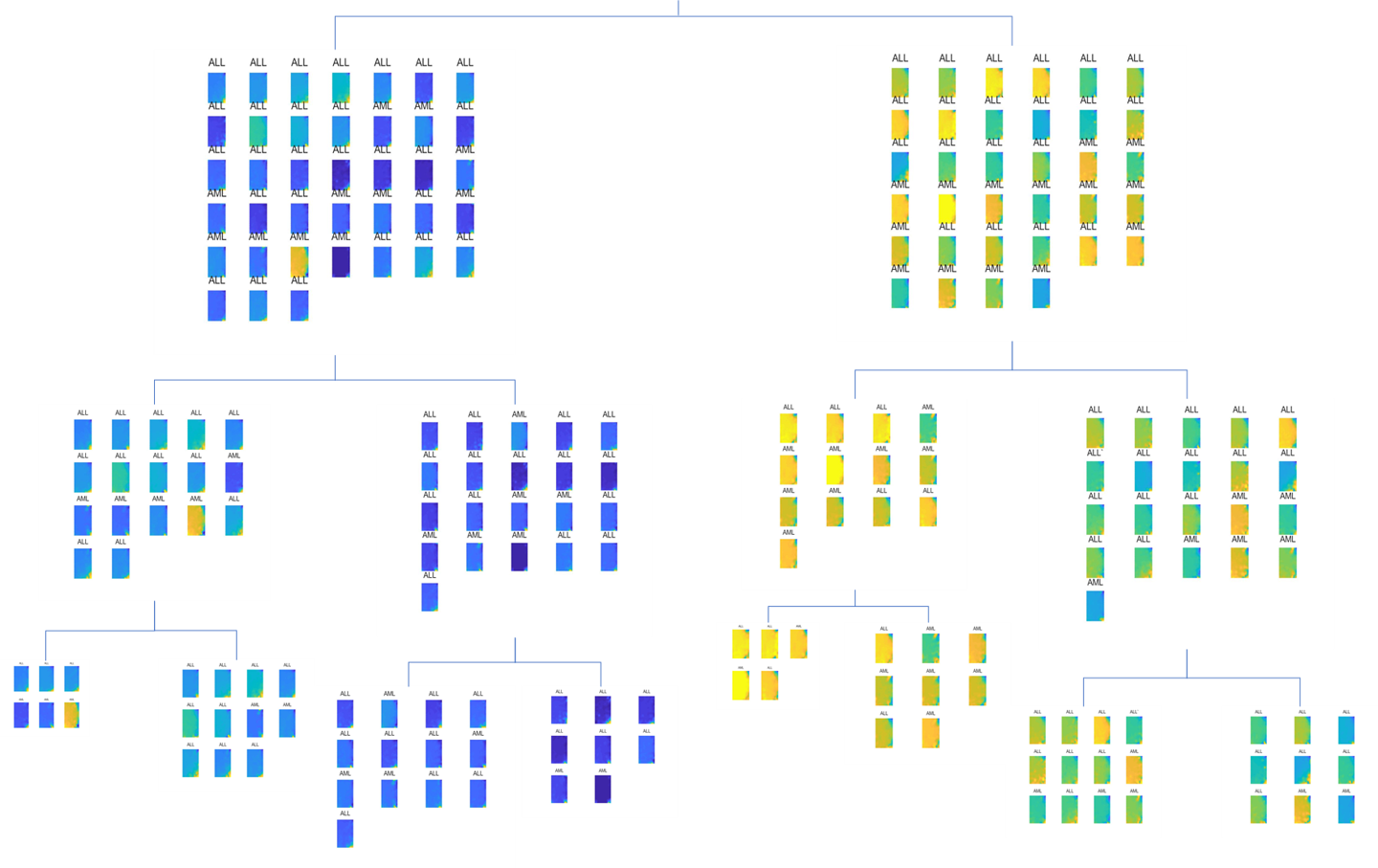}
    \caption{\label{fig:data1} (a) CAMDA'03 competition dataset and (b) Kaggle cancer data clustering using PPP}
\end{figure*} 

We showed that PPP enables k-means and VQ to find vectors in instance space which offer the best 'separability' for the sample space and increases 'clusterability' of a dataset. This mechanism converges to same settings irrespective of the initial conditions. Figure~\ref{fig:PPP_GMM2} is visual representation of this in two dimensions. The leafs of the binary tree thus built give final cluster structure of the data, thus eliminating the need to specify number of clusters in advance. This binary-tree can be cut at a desirable level.

The VQ on instance space also enables reducing the size of a dataset to manageable subset, given by $\Gamma_0$ (the granularity of which can be decided by the computational resources at hand), making this mechanism suitable for very large datasets.
\begin{cor}
Not all datasets, or problem sets can have such a discriminating set. But the datasets for which $\Gamma_1$ and $\Gamma_2$ (corresponding to optimal $\Phi$) can be found, the Llyod's algorithm can be used to partition using this set without getting stuck in a local minima.  
\end{cor}
\begin{figure}[ht]
\centering
     \includegraphics[height=5cm]{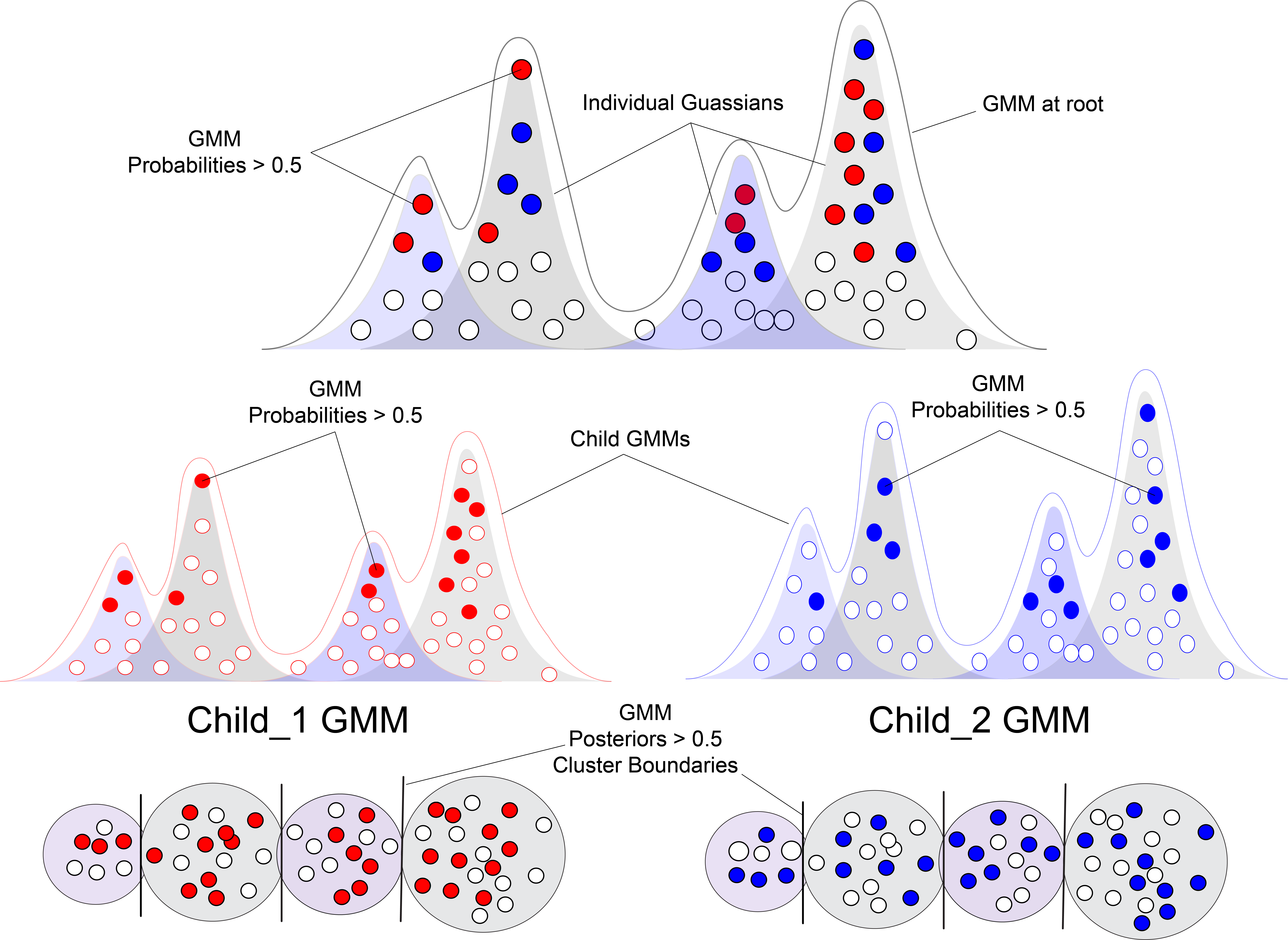}
     \caption{\label{fig:PPP_GMM2} PPP visualization for a hypothetical two-dimensional problem}
\end{figure}


\begin{thebibliography}{10}

\bibitem{ackerman2009clusterability}
Margareta Ackerman and Shai Ben-David.
\newblock Clusterability: A theoretical study.
\newblock In {\em Artificial Intelligence and Statistics}, pages 1--8, 2009.

\bibitem{arthur2007k}
David Arthur and Sergei Vassilvitskii.
\newblock k-means++: The advantages of careful seeding.
\newblock In {\em Proceedings of the eighteenth annual ACM-SIAM symposium on
  Discrete algorithms}, pages 1027--1035. Society for Industrial and Applied
  Mathematics, 2007.

\bibitem{bahmani2012scalable}
Bahman Bahmani, Benjamin Moseley, Andrea Vattani, Ravi Kumar, and Sergei
  Vassilvitskii.
\newblock Scalable k-means++.
\newblock {\em Proceedings of the VLDB Endowment}, 5(7):622--633, 2012.

\bibitem{balcan2009approximate}
Maria-Florina Balcan, Avrim Blum, and Anupam Gupta.
\newblock Approximate clustering without the approximation.
\newblock In {\em Proceedings of the twentieth annual ACM-SIAM symposium on
  Discrete algorithms}, pages 1068--1077. Society for Industrial and Applied
  Mathematics, 2009.

\bibitem{berkhin2006survey}
Pavel Berkhin.
\newblock A survey of clustering data mining techniques.
\newblock In {\em Grouping multidimensional data}, pages 25--71. Springer,
  2006.

\bibitem{cristianini2000introduction}
Nello Cristianini and John Shawe-Taylor.
\newblock {\em An introduction to support vector machines and other
  kernel-based learning methods}.
\newblock Cambridge university press, 2000.

\bibitem{dirichlet1850reduction}
G~Lejeune Dirichlet.
\newblock {\"U}ber die reduction der positiven quadratischen formen mit drei
  unbestimmten ganzen zahlen.
\newblock {\em Journal f{\"u}r die reine und angewandte Mathematik},
  40:209--227, 1850.

\bibitem{fahad2014survey}
Adil Fahad, Najlaa Alshatri, Zahir Tari, Abdullah Alamri, Ibrahim Khalil,
  Albert~Y Zomaya, Sebti Foufou, and Abdelaziz Bouras.
\newblock A survey of clustering algorithms for big data: Taxonomy and
  empirical analysis.
\newblock {\em IEEE transactions on emerging topics in computing},
  2(3):267--279, 2014.

\bibitem{forgy1965cluster}
Edward~W Forgy.
\newblock Cluster analysis of multivariate data: efficiency versus
  interpretability of classifications.
\newblock {\em biometrics}, 21:768--769, 1965.

\bibitem{friedman2001elements}
Jerome Friedman, Trevor Hastie, and Robert Tibshirani.
\newblock {\em The elements of statistical learning}, volume~1.
\newblock Springer series in statistics Springer, Berlin, 2001.

\bibitem{gersho2012vector}
Allen Gersho and Robert~M Gray.
\newblock {\em Vector quantization and signal compression}, volume 159.
\newblock Springer Science \& Business Media, 2012.

\bibitem{gray1998quantization}
Robert~M. Gray and David~L. Neuhoff.
\newblock Quantization.
\newblock {\em IEEE transactions on information theory}, 44(6):2325--2383,
  1998.

\bibitem{jolliffe2002principal}
Ian Jolliffe.
\newblock {\em Principal component analysis}.
\newblock Wiley Online Library, 2002.

\bibitem{kohonen2013essentials}
Teuvo Kohonen.
\newblock Essentials of the self-organizing map.
\newblock {\em Neural networks}, 37:52--65, 2013.

\bibitem{kohonen1998self}
Teuvo Kohonen and Panu Somervuo.
\newblock Self-organizing maps of symbol strings.
\newblock {\em Neurocomputing}, 21(1):19--30, 1998.

\bibitem{lloyd1982least}
Stuart Lloyd.
\newblock Least squares quantization in pcm.
\newblock {\em IEEE transactions on information theory}, 28(2):129--137, 1982.

\bibitem{mangiameli1996comparison}
Paul Mangiameli, Shaw~K Chen, and David West.
\newblock A comparison of som neural network and hierarchical clustering
  methods.
\newblock {\em European Journal of Operational Research}, 93(2):402--417, 1996.

\bibitem{ostrovsky2006effectiveness}
Rafail Ostrovsky, Yuval Rabani, Leonard~J Schulman, and Chaitanya Swamy.
\newblock The effectiveness of lloyd-type methods for the k-means problem.
\newblock In {\em Foundations of Computer Science, 2006. FOCS'06. 47th Annual
  IEEE Symposium on}, pages 165--176. IEEE, 2006.

\bibitem{ostrovsky2012effectiveness}
Rafail Ostrovsky, Yuval Rabani, Leonard~J Schulman, and Chaitanya Swamy.
\newblock The effectiveness of lloyd-type methods for the k-means problem.
\newblock {\em Journal of the ACM (JACM)}, 59(6):28, 2012.

\bibitem{roweis1998algorithms}
Sam Roweis.
\newblock Em algorithms for pca and spca.
\newblock {\em Advances in neural information processing systems}, pages
  626--632, 1998.

\bibitem{roweis2000nonlinear}
Sam~T Roweis and Lawrence~K Saul.
\newblock Nonlinear dimensionality reduction by locally linear embedding.
\newblock {\em Science}, 290(5500):2323--2326, 2000.

\bibitem{sturn2002genesis}
Alexander Sturn, John Quackenbush, and Zlatko Trajanoski.
\newblock Genesis: cluster analysis of microarray data.
\newblock {\em Bioinformatics}, 18(1):207--208, 2002.

\bibitem{tayal2014primal}
Aditya Tayal, Thomas~F Coleman, and Yuying Li.
\newblock Primal explicit max margin feature selection for nonlinear support
  vector machines.
\newblock {\em Pattern Recognition}, 47(6):2153--2164, 2014.

\bibitem{vesanto1999som}
Juha Vesanto.
\newblock Som-based data visualization methods.
\newblock {\em Intelligent data analysis}, 3(2):111--126, 1999.

\bibitem{vesanto1999self}
Juha Vesanto, Johan Himberg, Esa Alhoniemi, Juha Parhankangas, et~al.
\newblock Self-organizing map in matlab: the som toolbox.
\newblock In {\em Proceedings of the Matlab DSP conference}, volume~99, pages
  16--17, 1999.

\bibitem{voronoi1908nouvelles}
Georges Vorono{\"\i}.
\newblock Nouvelles applications des param{\`e}tres continus {\`a} la
  th{\'e}orie des formes quadratiques. deuxi{\`e}me m{\'e}moire. recherches sur
  les parall{\'e}llo{\`e}dres primitifs.
\newblock {\em Journal f{\"u}r die reine und angewandte Mathematik},
  134:198--287, 1908.

\bibitem{wu2012advances}
Junjie Wu.
\newblock {\em Advances in K-means clustering: a data mining thinking}.
\newblock Springer Science \& Business Media, 2012.

\bibitem{xu2005survey}
Rui Xu and Donald Wunsch.
\newblock Survey of clustering algorithms.
\newblock {\em IEEE Transactions on neural networks}, 16(3):645--678, 2005.

\end{thebibliography}

\end{document}